\documentclass[a4paper,UKenglish,cleveref, autoref, thm-restate]{oasics-v2021}

\usepackage{amsmath}
\usepackage{amssymb}
\usepackage{amsthm}
\usepackage{mathrsfs}
\usepackage{wasysym}
\usepackage{bbm}		
\usepackage{stmaryrd}			
\usepackage{braket}		
\usepackage{cancel}		
\usepackage{nicefrac}
\usepackage{mathtools}
\usepackage{autonum}						

\newcommand{\ZZ}{\mathbb{Z}}			
\newcommand{\NN}{\mathbb{N}}			

\newcommand{\isdef}{\triangleq}			
\DeclarePairedDelimiter\abs{\lvert}{\rvert}		
\DeclarePairedDelimiter\norm{\lVert}{\rVert}	
\DeclarePairedDelimiter\intinterval{\llbracket}{\rrbracket}		

\newcommand{\ee}{\mathrm{e}}			
\newcommand{\bigo}{O}					
\newcommand{\xPr}{\operatorname{\mathbb{P}}}		
\newcommand\given[1][]{\:#1\vert\:}					

\newcommand{\pspace}[1]{\mathcal{#1}}	
\newcommand{\family}[1]{\mathscr{#1}}	
\newcommand{\TV}{\mathrm{TV}}			
\newcommand{\mix}{\mathsf{mix}}			
\newcommand{\hmax}{
	\overline{h}%
}
\newcommand{\distUnif}{\mathsf{Uniform}}	
\newcommand{\distBern}{\mathsf{Bernoulli}}	
\newcommand{\Neighb}{\mathcal{N}}			
\newcommand{\Moore}{\mathcal{M}}			

\newcommand{\xclass}[1]{\mathbf{#1}}	

\graphicspath{{./figures/}}

\title{
	Reversible cellular automata in presence of noise rapidly forget everything\footnote{%
		To appear in the Proceedings of \href{https://automata2021.lis-lab.fr/}{AUTOMATA 2021}.
		Last update: April 23, 2021.
	}
} 

\titlerunning{Reversible CA with noise} 

\author{Siamak Taati}{Department of Mathematics, American University of Beirut, Lebanon}{siamak.taati@gmail.com}{https://orcid.org/0000-0002-6503-2754}{}

\authorrunning{S. Taati} 

\Copyright{Siamak Taati} 



\category{Invited Talk} 

\relatedversion{} 

\nolinenumbers 

\hideOASIcs 

\EventEditors{Alonso Castillo-Ramirez, Pierre Guillon, and K\'{e}vin Perrot}
\EventNoEds{3}
\EventLongTitle{27th IFIP WG 1.5 International Workshop on Cellular Automata and Discrete Complex Systems (AUTOMATA 2021)}
\EventShortTitle{AUTOMATA 2021}
\EventAcronym{AUTOMATA}
\EventYear{2021}
\EventDate{July 12--14, 2021}
\EventLocation{Aix-Marseille University, France}
\EventLogo{}
\SeriesVolume{90}
\ArticleNo{3}

\begin{document}

\maketitle

\begin{abstract}
	We consider reversible and surjective cellular automata perturbed with noise.
	We show that, in the presence of positive additive noise,
	the cellular automaton forgets all the information regarding its initial configuration
	exponentially fast.
	In particular, the state of a finite collection of cells with diameter~$n$ becomes
	indistinguishable from pure noise after $\bigo(\log n)$ time steps.
	This highlights the seemingly unavoidable need for irreversibility
	in order to perform scalable reliable computation in the presence of noise.
\end{abstract}

\section{Introduction}
\subsection{Background}

A major challenge regarding the physical implementation of computation
is the inevitability of transient errors due to noise.
The difficulty is that, even if each component of the system is built to be highly accurate
and has a very small probability of error, in a lengthy computation, occasional errors
are bound to occur.  Such errors may then propagate and entirely corrupt the computation.
The problem of how to perform computation reliably in the presence of noise
(via suitable error-correcting mechanisms)
goes back to von~Neumann and has been studied since~\cite{Neu56,Gac05}.

At the nanoscopic scale, the issue of thermal noise becomes even more pressing.
Not only are the nano-scale components more sensitive to any sort of fluctuations,
but also the heat generated by the computational process has little time and space
to escape the system, and thus leads to an increase in thermal noise.
Landauer argued that the heat generated by a computational process
is associated to its logical irreversibility, and identified a theoretical
lower bound for the amount of heat dissipated in the process of erasing
a bit of information~\cite{Lan61}.
Bennett, Fredkin and Toffoli showed that, at least in theory, any computation can
be efficiently simulated by a logically reversible one, hence requiring
virtually no dissipation~\cite{Ben73, Ben82, FreTof82, Ben89}.
The output of the reversible computation will consist of the intended output
as well as some extra information that allows one to trace the computation
backwards.
If need be, this extra information can be erased away from the computer core,
hence avoiding the accumulation of heat.

By the virtue of their ``physics-like'' features,
cellular automata (CA) have been a popular mathematical
model for studying the physical aspects of computation.
The notion of reversibility has a natural formulation in the setting of CA,
and reversible CA have been widely studied,
not only as models of reversible computers, but also as models of
physical processes and from other points of
view~\cite{Tof77,Vic84,Mar84,TofMar87,MorHar89,TofMar90,Kar94,ChoDro98,Mor08,Sch15,SalTor17,Kar18}.
The reliability of computation in the presence of noise is also
studied in the setting of
cellular automata~\cite{Too80, Gac86, GacRei88, BerSim88, BraNeu94, Gac01, McCPip08, GacTor18, MarSabTaa19}.

Although logical reversibility solves the issue of heat generation
in a computational process, it leads to another difficulty.
Even if it does not dissipate heat itself,
the process is still exposed to external noise.
This external noise can potentially be reduced with proper insulation
but can never be eliminated altogether.
The logical reversibility of the process entails that the noise entering the system
is not dissipated and hence accumulates inside the system.  This means that,
unless one finds a clever workaround,
the state of the system will eventually be overcome by noise~\cite{Ben82}.

For a model of noisy reversible circuits,
Aharonov, Ben-Or, Impagliazzo and Nisan~\cite{AhaBenImpNis96} identified
the limitation imposed by the accumulation of noise.
They considered reversible circuits in which errors occur on wires
at each ``time unit'',
and proved that the output of such a circuit is indistinguishable from pure noise
unless the size of the circuit is exponential in its depth.
Conversely, they showed that every classical Boolean circuit with size~$s$ and depth~$d$
can be simulated by a noisy reversible circuit with size $\bigo\big(s\times 2^{\bigo(d)}\big)$
and depth~$\bigo(d)$.
In particular, polynomial-size noisy reversible circuits
have the power of the complexity class~$\xclass{NC}^1$.
They also proved a similar (though not as sharp) result concerning
noisy quantum circuits.

The result we present here formulates a similar limitation
imposed by the accumulation of noise in the setting of cellular automata.
We show that a reversible CA subject to positive additive noise
forgets its initial configuration exponentially fast, in the sense that
the state of any finite collection of its cells with diameter~$n$ becomes
indistinguishable from pure noise after $\bigo(\log n)$ number of time steps.
It remains open whether any meaningful computation can be done
reliably with such a limitation.

Mathematically, the forgetfulness of a reversible CA subject to noise
corresponds to the exponential ergodicity of the resulting probabilistic CA.
This means that the distribution of the process converges exponentially fast
to a unique invariant measure, which in this case is the uniform Bernoulli measure
(i.e., the distribution of a configuration chosen by independent coin flips).
This result improves upon an earlier partial ergodicity result~\cite{MarSabTaa19},
which was limited
to shift-invariant initial measures, and in which the rate of convergence
was not identified.  The exponential ergodicity of noisy reversible CA
is a special case of a more general ergodicity result
concerning positive-rate probabilistic CA with Bernoulli invariant measures
and their asynchronous counterparts~\cite{MarTaa21}.

The structure of the paper is as follows.
In \cref{sec:setting}, we introduce the setting, and in \cref{sec:theorem:main}
we state the main result.
The proof of the main result, which is based on entropy,
appears in \cref{sec:entropy}.
\Cref{sec:information-loss} is dedicated to the interpretation
of the result regarding the rapid information loss of reversible computers
in the presence of noise.
The paper is concluded with some discussions in \cref{sec:discussion}.

\subsection{Setting}
\label{sec:setting}

\paragraph*{General notation}
We will use the notation $\NN\isdef\{0,1,2,\ldots\}$ and $\ZZ^+\isdef\{1,2,3,\ldots\}$.
We will write $\intinterval{n,m}$ to denote the integer interval~$\{n,n+1,\ldots,m\}$.
We use the notation $x_A$ for the restriction of a function $x$ to a subset $A$
of its domain.
We write $Z\sim q$ to indicate that $Z$ is a random variable with distribution~$q$.
The total variation distance between two probability distributions $p$ and $q$ on a finite set~$A$
is
\begin{align}
	\norm{q-p}_\TV &\isdef
		\sup_{E\subseteq A}\abs{q(E)-p(E)} = \frac{1}{2}\sum_{a\in A}\abs{q(a)-p(a)} \;.
\end{align}
Throughout this article, $\log(\cdot)$ stands for the natural logarithm.

\paragraph*{Cellular automata}
Cellular automata (CA) are abstract models of massively parallel computation.
A \emph{configuration} of the model is an assignment of symbols from a finite alphabet~$\Sigma$
to every site of the lattice $\ZZ^d$ (for $d=1,2,\ldots$).
The sites of the lattice are called \emph{cells}
and the symbol on each cell is referred to as its \emph{state}.
At each step of the computation,
the states of all cells are simultaneously updated according to a \emph{local rule}.
The local rule takes into account the current state of the cell to be updated
as well as its neighbours.
More specifically,
the local rule is a function $f:\Sigma^\Neighb\to\Sigma$, where $\Neighb\subseteq\ZZ^d$ is
a finite set indicating the relative positions of the \emph{neighbours} of each cell,
possibly including the cell itself.
A configuration $x\in\Sigma^{\ZZ^d}$ is updated to a configuration $Fx\in\Sigma^{\ZZ^d}$,
where $(Fx)_i \isdef f\big((x_{i+a})_{a\in\Neighb}\big)$ for each cell $i\in\ZZ^d$.
We refer to~$F$ as the \emph{global map} of the CA.
The computation thus consists in iterating the global map $F$ on an initial configuration.
We identify a CA with its global map~$F$ and speak of the CA $F$.

\paragraph*{Surjectivity, injectivity, and reversibility}
A CA is said to be \emph{surjective} (resp., \emph{injective}, \emph{bijective})
if its global map is surjective (resp., injective, bijective).
It is well-known that every injective CA is also surjective,
and hence bijective.
In fact, the Garden-of-Eden theorem states that surjectivity
is equivalent to pre-injectivity~\cite{Moo62,Myh63}.
A CA $F$ is said to be \emph{pre-injective} if whenever two distinct configurations $x$ and $y$
agree on all but finitely many cells, their images $Fx$ and $Fy$ are distinct.

A CA $F$ is said to be \emph{reversible}
if $F$ is an invertible map and $F^{-1}$ is itself a CA.
It follows from a topological argument that every bijective CA
is automatically reversible~\cite{Hed69}.  Thus, injectivity, bijectivity and reversibility
are equivalent conditions.

\paragraph*{Noise}
In this paper, we are concerned with the effect of transient noise
on the computation carried out by a CA.
In the presence of noise, random errors might occur during the updates
of the cells.

We restrict ourselves to a specific model of noise,
namely \emph{additive noise}, although the results
of the current paper remain true with the somewhat more general
model of \emph{permutation noise}~\cite{MarSabTaa19}.
We (arbitrarily) identify~$\Sigma$ with an Abelian group $(\Sigma,+)$.
Subject to an additive noise with \emph{noise distribution}~$q$,
a symbol $a$ is replaced with $a+Z$, where $Z$ is a random variable
with distribution~$q$.  We will assume that the noise distribution~$q$
is strictly positive.  In particular,
$\xPr(a+Z=b)=q(b-a)>0$ for every $a,b\in\Sigma$.

At each time step of the computation, the state of each cell is
first updated according to the local rule of the CA and is
then perturbed with positive additive noise. 
The noise variables at different cells and different time steps
are all assumed to be independent.

More specifically, the noise is described by a family $(Z^t_i)_{i\in\ZZ^d,t\in\ZZ^+}$
of independent random variables with distribution~$q$.
The trajectory of the noisy computation starting from a configuration~$x$
is given by a sequence of random configurations $(X^t)_{t\in\NN}$,
where $X^0\isdef x$, and 
\begin{align}
	X^t_i &\isdef f\big((X^{t-1}_{i+a})_{a\in\Neighb}\big) + Z^t_i
\end{align}
for every time step $t>0$ and every cell $i\in\ZZ^d$.

\paragraph*{Probabilistic CA}
The noisy computation can be described by a probabilistic CA.
In a \emph{probabilistic} CA (PCA), the local rule is probabilistic and
the updates at different cells and different time steps are performed independently.
More specifically, the local rule is given by a stochastic matrix
$\varphi:\Sigma^\Neighb\times\Sigma\to[0,1]$
(hence, $\sum_{b\in\Sigma}\varphi(u,b)=1$ for each $u\in\Sigma^\Neighb$).
A configuration $x$ is updated to a random configuration $Y$
with distribution
\begin{align}
	\xPr(Y_A=y_A) = \prod_{i\in A}\varphi\big((x_{i+a})_{a\in\Neighb},y_i\big)
\end{align}
for every finite set $A\subseteq\ZZ^d$.
The role of the global map in the deterministic case is played by a (global) transition kernel $\Phi$
where $\Phi(x,\cdot)$ indicates the distribution of $Y$.
(See~\cite{TooVasStaMitKurPir90} or~\cite{MarSabTaa19} for more details.)
The trajectory of a PCA $\Phi$ is a Markov process with transition kernel~$\Phi$,
that is, a sequence $(X^t)_{t\in\NN}$ of random configurations such that
\begin{romanenumerate}
	\item Given $X^t$, the configuration $X^{t+1}$ is independent of
		the configurations $X^0,X^1,\ldots,X^{t-1}$.
	\item Given $X^t$, the distribution of $X^{t+1}$ is given by $\Phi(X^t,\cdot)$.
\end{romanenumerate}

In the case of a CA $F$ with additive noise, the local rule of the resulting PCA
is given by
\begin{align}
	\varphi(u,b) &\isdef q\big(b-f(u)\big) \;,
\end{align}
where $f$ is the local rule of $F$ and $q$ is the noise distribution. 

\paragraph*{Ergodicity}

We say that a probability measure $\lambda$ on $\Sigma^{\ZZ^d}$ is \emph{invariant} under
a PCA~$\Phi$ if $X^{t+1}\sim\lambda$ whenever $X^t\sim\lambda$.
We say that $\Phi$ is \emph{ergodic} if it has a unique invariant measure $\lambda$
and furthermore, for any (possibly random) starting configuration~$X^0$, the distribution of $X^t$
converges weakly to $\lambda$.  This means that for every
finite set $A\subseteq\ZZ^d$ and every $u\in\Sigma^A$,
\begin{align}
	\xPr(X^t_A=u) &\to \lambda\big(\{\hat{x}:\hat{x}_A=u\}\big)
		\qquad\text{as $t\to\infty$.}
\end{align}
We can interpret the ergodicity of a PCA $\Phi$
as $\Phi$ ``forgetting'' its initial configuration.
However, the convergence (and hence the process of forgetting the initial configuration)
can potentially be slow.

Among the PCA that are ergodic, it is quite common that the convergence
towards the unique invariant measure is exponentially fast, in the sense that,
for every finite set $A\subseteq\ZZ^d$,
\begin{align}
	\norm[\big]{\xPr(X^t_A\in\cdot\,)-\lambda\big(\{\hat{x}:\hat{x}_A\in\cdot\,\}\big)}_\TV
		&\leq \alpha_A\ee^{-\beta t} \;,
\end{align}
where $\beta>0$ is a constant (independent of~$A$) and $\alpha_A$ depends on the set~$A$
but not on~$t$ or the initial configuration~$X^0=x$.
(Here, $\xPr(X^t_A\in\cdot\,)$ stands for the distribution of $X^t_A$
and $\lambda\big(\{\hat{x}:\hat{x}_A\in\cdot\,\}\big)$ for the marginal of~$\lambda$
on~$A$.)
It is the dependence of $\alpha_A$ on the set~$A$ that has more relevant information
on the speed of convergence.
We will discuss this further in \cref{sec:information-loss}.

In the current paper,
the unique invariant measure of the PCA we study will be the \emph{uniform Bernoulli measure},
that is, the distribution of a random configuration
in which the states of different cells are chosen uniformly at random from~$\Sigma$
and independently from one another.

\subsection{Statement of the theorem}
\label{sec:theorem:main}

While the primary interest here is in reversible CA in the presence of noise,
our main result remains true for surjective CA.
By the \emph{diameter} of a finite set $A\subseteq\ZZ^d$
we mean the smallest $n\in\NN$ such that $A\subseteq u+\intinterval{0,n-1}^d$
for some $u\in\ZZ^d$.

\begin{theorem}[Surjective CA with additive noise]
\label{thm:main-result}
	Let $F:\Sigma^{\ZZ^d}\to\Sigma^{\ZZ^d}$ be a surjective CA.
	Let $\Phi$ be a PCA describing the perturbation of $F$ with positive additive noise.
	Then, $\Phi$ is exponentially ergodic with the uniform Bernoulli measure~$\lambda$
	as the unique invariant measure.
	
	More specifically, there exist constants $\alpha,\beta,a,b>0$
	such that if $(X^t)_{t\in\NN}$ is a trajectory of $\Phi$ with arbitrary initial configuration,
	then
	\begin{align}
		\norm[\big]{\xPr(X^t_A\in\cdot\,)-\lambda\big(\{\hat{x}:\hat{x}_A\in\cdot\,\}\big)}_\TV
			&\leq \alpha\ee^{-\beta t}n^{(d-1)/2} \;,
	\end{align}
	for every finite set $A\subseteq\ZZ^d$ with diameter~$n$
	and every $t\geq a\log n + b$.
\end{theorem}

The above theorem completes an earlier partial result in which the uniqueness and convergence
(without rate of convergence) was established only among shift-invariant measures~\cite{MarSabTaa19}.
Theorem~\ref{thm:main-result} is a special case (relevant to reversible computing)
of a more general result: every probabilistic CA with strictly positive transition probabilities
which has a Bernoulli invariant measure is exponentially ergodic~\cite{MarTaa21}.

\section{Accumulation of entropy}
\label{sec:entropy}

Like the result of Aharonov et al.~\cite{AhaBenImpNis96},
the proof of \cref{thm:main-result} is based on entropy.
Recall that the (Shannon) \emph{entropy} of a discrete random variable $X$ (measured in \emph{nats})
is defined as
\begin{align}
	H(X) &\isdef -\sum_x\xPr(X=x)\log\xPr(X=x) \;.
\end{align}
The entropy of~$X$ measures the average information content of~$X$.
If $X$ takes its values within a finite set $\Gamma$,
then $H(X)\leq\log\abs{\Gamma}$, with equality if and only if~$X$ is uniformly
distributed over~$\Gamma$.
We refer to~\cite{CovTho06} for information on entropy and its properties.

We follow the approach of the earlier proof of ergodicity modulo shift~\cite{MarSabTaa19}.
Namely, we use the fact that positive additive noise increases the entropy
of every finite collection of cells, while a surjective CA does not erase
the entropy and only ``diffuses'' it.
In order to achieve complete ergodicity with sharp rate of convergence,
we use two new ingredients:
\begin{alphaenumerate}
	\item An explicit bound on the amount of entropy increase due to noise,
	\item A ``bootstrap argument'' showing that if the rate of entropy increase
		is high compared to the rate of entropy leakage, then the entropy of each
		finite set will inevitably accumulate and rise up to
		its maximum capacity.
\end{alphaenumerate}

That the proof is based on entropy is natural.
A reversible CA in the presence of noise can be thought of as
a (microscopically reversible) physical system in contact
with a heat bath.  The entropy increase is therefore a manifestation
of the second law of thermodynamics for such systems.

\subsection{Effect of additive noise on entropy}
\label{sec:entropy:noise}

We set $\hmax\isdef\log\abs{\Sigma}$, so that $\hmax$ is the highest possible
entropy of a $\Sigma$-valued random variable.
Throughout this section, we also let $q:\Sigma\to(0,1)$ be a strictly positive probability distribution,
and set $\kappa\isdef\abs{\Sigma}\min_{a\in\Sigma}q(a)$.
Note that $0<\kappa\leq 1$.
To avoid trivial situations, we assume that $\Sigma$ has at least two elements
and that $q$ is not uniform.  Hence, $\hmax>0$ and $\kappa<1$.

\begin{lemma}[Effect of additive noise on entropy]
\label{lem:entropy:noise}
	If $A$ and $N$ are independent $\Sigma$-valued random variables with $N\sim q$,
	then
	\begin{align}
		H(A+N) &\geq \kappa\hmax + (1-\kappa) H(A) \;.
	\end{align}
\end{lemma}
\begin{proof}
	The distribution $q$ can be decomposed as
	\begin{align}
		q &= \kappa u + (1-\kappa)\tilde{q}
	\end{align}
	where $u$ is the uniform distribution on $\Sigma$ and $\tilde{q}$
	is another distribution on $\Sigma$ given by
	$\tilde{q}(a)\isdef\frac{q(a)-\nicefrac{\kappa}{\abs{\Sigma}}}{1-\kappa}$ for $a\in\Sigma$.
	Thus, without loss of generality (by defining a new probability space if necessary),
	we can assume that $N=BU+(1-B)\tilde{N}$ where $B\sim\distBern(\kappa)$, $U\sim\distUnif(\Sigma)$
	and $\tilde{N}\sim\tilde{q}$, and the variables $A$, $B$, $U$ and $\tilde{N}$ are independent.
	
	Using this representation, we have
	\begin{align}
		\MoveEqLeft H(A+N) \\
		&= H(A+BU+(1-B)\tilde{N}) \\
		&\geq H(A+BU+(1-B)\tilde{N}\given B) \\
		&= \kappa H(A+\underbrace{BU + (1-B)\tilde{N}}_{\text{$=U$ when $B=1$}}\given B=1)
			+ (1-\kappa)H(A+\underbrace{BU+(1-B)\tilde{N}}_{\text{$=\tilde{N}$ when $B=0$}}\given B=0) \\
		&=
			\kappa\hmax + (1-\kappa)H(A+\tilde{N}) \;.
	\end{align}
	The claim follows once we recall that $H(A+\tilde{N})\geq H(A)$
	whenever $A$ and $\tilde{N}$ are independent $\Sigma$-valued random variables.
	Namely, we have
	\begin{align}
		H(A+\tilde{N},\tilde{N}) &= H(\tilde{N}) + \overbrace{H(A+\tilde{N}\given \tilde{N})}^{=H(A)} \\
		H(A+\tilde{N},\tilde{N}) &= H(A+\tilde{N}) + H(\tilde{N}\given A+\tilde{N})
	\end{align}
	from which we get $H(A+\tilde{N}) = H(A) + I(A+\tilde{N};\tilde{N})$
	where $I(A+\tilde{N};\tilde{N})\geq 0$ is the mutual information between $A+\tilde{N}$ and $\tilde{N}$.
\end{proof}

The conditional version of the above lemma can be proven similarly,
or by reducing it to the unconditional version.
\begin{lemma}[Effect of additive noise on conditional entropy]
\label{lem:entropy:noise:conditional}
	If $A$ and $N$ are $\Sigma$-valued random variables and $C$
	is another random variable conditioned on which $A$ and $N$ are independent
	with $N\sim q$, then
	\begin{align}
		H(A+N\given C) &\geq \kappa\hmax + (1-\kappa) H(A\given C) \;.
	\end{align}
\end{lemma}

For a collection of random symbols subjected to independent noise,
we have the following lemma as a corollary.
\begin{lemma}[Effect of additive noise on joint entropy]
\label{lem:entropy:noise:joint}
	If $\underline{A}\isdef(A_1,A_2,\ldots,A_n)$ and $\underline{N}\isdef(N_1,N_2,\ldots,N_n)$
	are two independent collections of $\Sigma$-valued random variables
	and $N_1,N_2,\allowbreak\ldots,\allowbreak N_n$ are i.i.d.\ with distribution $q$, then
	\begin{align}
		H(\underline{A}+\underline{N}) &\geq n\kappa\hmax + (1-\kappa)H(\underline{A}) \;.
	\end{align}
\end{lemma}
\begin{proof}
	Using the chain rule and the fact that $N_k$'s are independent of one another
	and	independent $A_k$'s, we have
	\begin{align}
		H(\underline{A}) &=	\label{eq:entropy:noise:joint:chain:1}
		\sum_{k=1}^n H\big(A_k \given[\big] (A_i)_{i<k}\big)
	\shortintertext{and}
		H(\underline{A}+\underline{N}) &=
			\sum_{k=1}^n H\big(A_k+N_k \given[\big] (A_i+N_i)_{i<k}\big) \\
		&\geq
			\sum_{k=1}^n H\big(A_k+N_k \given[\big] (A_i)_{i<k}, (N_i)_{i<k}\big) \\
		&= \label{eq:entropy:noise:joint:chain:2}
			\sum_{k=1}^n H\big(A_k+N_k \given[\big] (A_i)_{i<k}\big) \;.
	\end{align}
	Applying \cref{lem:entropy:noise:conditional}
	to the corresponding terms in~\eqref{eq:entropy:noise:joint:chain:1}
	and~\eqref{eq:entropy:noise:joint:chain:2} yields the result.
\end{proof}

\subsection{Effect of a surjective CA on entropy}
\label{sec:entropy:surjective-CA}

The effect of surjective CA on entropy was clarified in the earlier proof
of ergodicity modulo shift~\cite{MarSabTaa19}.
For completeness, we recall the proof.
Given a set $J\subseteq\ZZ^d$ and an integer $r\geq 0$,
we denote by $\Moore^r(J)\isdef J+\intinterval{-r,r}^d$ the set of all cells
that are within distance $r$ from~$J$.
We also define $\partial\Moore^r(J)\isdef\Moore^r(J)\setminus J$.
\begin{lemma}[Effect of a surjective CA on entropy]
\label{lem:entropy:surjective-ca}
	Let $F:\Sigma^{\ZZ^d}\to\Sigma^{\ZZ^d}$ be a surjective CA
	with neighbourhood $\Neighb\subseteq\intinterval{-r,r}$.
	Then, for every random configuration $X$ and every finite set $J\subseteq\ZZ^d$
	we have
	\begin{align}
		H\big((FX)_J\big) &\geq H(X_J) - c(J)
	\end{align}
	where $c(J)\isdef\big(\abs[\big]{\partial\Moore^{2r}(J)}+\abs[\big]{\partial\Moore^r(J)}\big)\hmax$.
\end{lemma}
\begin{proof}
	By the Garden-of-Eden theorem, $F$ is pre-injective~\cite{Moo62,Myh63}.
	From the pre-injectivity of~$F$ it follows that for every configuration $x$,
	the pattern $x_J$ is uniquely determined from the patterns $x_{\partial\Moore^{2r}(J)}$
	and $(Fx)_{\Moore^r(J)}$.
	Indeed, suppose that $x$ and $x'$ are two configurations such that
	$x'_{\partial\Moore^{2r}(J)}=x_{\partial\Moore^{2r}(J)}$ and
	$(Fx')_{\Moore^r(J)}=(Fx)_{\Moore^r(J)}$.
	Let $x''$ be another configuration that agrees with $x'$ on $\Moore^{2r}(J)$
	and with $x$ outside $J$.
	Then, $x$ and $x''$ agree everywhere except possibly on~$J$.
	On the other hand, $Fx=Fx''$ because
	\begin{itemize}
		\item $Fx''$ and $Fx$ agree outside $\Moore^r(J)$ because
			$x''$ and $x$ agree outside~$J$,
		\item $Fx''$ and $Fx'$ agree on $\Moore^r(J)$ because
			$x''$ and $x'$ agree on $\Moore^{2r}(J)$,
			and $Fx'$ and $Fx$ agree on $\Moore^r(J)$ by assumption.
	\end{itemize}
	The pre-injectivity of $F$ now implies $x''=x$.
	It follows that $x$ and $x'$ agree on~$J$.
	
	Now, consider a random configuration~$X$.
	Since $X_J$ is a function of 
	$X_{\partial\Moore^{2r}(J)}$ and $(FX)_{\Moore^r(J)}$,
	we have
	\begin{align}
		H(X_J) &\leq H\big(X_{\partial\Moore^{2r}(J)},(FX)_{\Moore^r(J)}\big) \\
		&=
			H\big((FX)_J\big) +
			H\big(X_{\partial\Moore^{2r}(J)},(FX)_{\partial\Moore^r(J)}\given[\big](FX)_J\big) \\
		&\leq
			H\big((FX)_J\big) + \big(\abs[\big]{\partial\Moore^{2r}(J)}+\abs[\big]{\partial\Moore^r(J)}\big)\hmax \;,
	\end{align}
	proving the claim.
\end{proof}

\subsection{Evolution of entropy}
\label{sec:entropy:evolution}

Combining the above lemmas we obtain the following proposition.
\begin{proposition}[Evolution of entropy]
\label{prop:entropy:evolution}
	Let $\Phi$ be a PCA describing the perturbation of a surjective CA $F:\Sigma^{\ZZ^d}\to\Sigma^{\ZZ^d}$
	with positive additive noise,
	and let $q:\Sigma\to(0,1)$ denote the noise distribution.
	Let $X^0,X^1,\ldots$ be a trajectory of $\Phi$ starting from an arbitrary
	configuration~$X^0$.
	Then, for every finite $J\subseteq\ZZ^d$ and each $t\geq 0$ we have
	\begin{align}
		H\big(X^t_J\big) &\geq
			[1-(1-\kappa)^t]\abs{J}\hmax - \tilde{c}(J)
	\end{align}
	where $\kappa\isdef\abs{\Sigma}\min_{a\in\Sigma}q(a)$
	and $\tilde{c}(J)\isdef
	\big(\frac{1-\kappa}{\kappa}\big)\big(\abs[\big]{\partial\Moore^{2r}(J)}+\abs[\big]{\partial\Moore^r(J)}\big)\hmax$.
\end{proposition}
\begin{proof}
	According to \cref{lem:entropy:surjective-ca}, for every $s>0$ we have
	\begin{align}
		H\big((FX^{s-1})_J\big) &\geq H(X^{s-1}_J) - c(J) \;.
	\end{align}
	\Cref{lem:entropy:noise:joint} on the other hand gives
	\begin{align}
		H(X^s_J) &\geq \kappa\abs{J}\hmax + (1-\kappa)H\big((FX^{s-1})_J\big) \;.
	\end{align}
	Combining the two, we find that for $s>0$,
	\begin{align}
	\label{eq:entropy:noisy-surjective-ca}
		H(X^s_J) &\geq (1-\kappa)H(X^{s-1}_J) + \kappa\abs{J}\hmax - (1-\kappa)c(J)\;.
	\end{align}
	Multiplying by $(1-\kappa)^{t-s}$, we obtain
	\begin{align}
		(1-\kappa)^{t-s}H(X^s_J)
			&\geq (1-\kappa)^{t-s+1}H(X^{s-1}_J) + \kappa(1-\kappa)^{t-s}\abs{J}\hmax - (1-\kappa)^{t-s+1}c(J) \;.
	\end{align}
	Summing over $s$ from $1$ to $t$, we get
	\begin{align}
		\sum_{s=1}^t (1-\kappa)^{t-s}H(X^s_J)
			&\geq
				\begin{multlined}[t]
					\sum_{s=1}^t(1-\kappa)^{t-s+1}H(X^{s-1}_J) \\
					+ \big[1-(1-\kappa)^t\big]\abs{J}\hmax
					- \big[\smash{\overbrace{1-(1-\kappa)^t}^{<1}}\big]\frac{1-\kappa}{\kappa}c(J)
				\end{multlined}
	\end{align}
	which after cancellation of the common terms gives
	\begin{align}
		H(X^t_J) &\geq
			(1-\kappa)^t H(X^0_J) + \big[1-(1-\kappa)^t\big]\abs{J}\hmax - \tilde{c}(J) \\
		&\geq
			\big[1-(1-\kappa)^t\big]\abs{J}\hmax - \tilde{c}(J) \;. \\
		& \qedhere
	\end{align}
\end{proof}

As a corollary, we get the following proposition.

\begin{proposition}[Evolution of entropy]
\label{prop:entropy:evolution:equilibrium}
	Let $\Phi$ be a PCA describing the perturbation of a surjective CA $F:\Sigma^{\ZZ^d}\to\Sigma^{\ZZ^d}$
	with positive additive noise,
	and let $q:\Sigma\to(0,1)$ denote the noise distribution.
	There are two constants $a_0,b_0>0$ with the following property.
	If $X^0,X^1,\ldots$ is a trajectory of~$\Phi$ starting from an arbitrary configuration $X^0$,
	then for every finite set $J\subseteq\ZZ^d$, we have
	\begin{align}
		H(X^t_J) &\geq \abs{J}\hmax - 2\tilde{c}(J)
		\qquad\text{for all $t\geq a_0\log\frac{\abs{J}}{\tilde{c}(J)} + b_0$}
	\end{align}
	where $\tilde{c}(J)\isdef
	\big(\frac{1-\kappa}{\kappa}\big)\big(\abs[\big]{\partial\Moore^{2r}(J)}+\abs[\big]{\partial\Moore^r(J)}\big)\hmax$.
\end{proposition}
\begin{proof}
	From \cref{prop:entropy:evolution},
	it follows that in order to have $H(X^t_J) \geq \abs{J}\hmax - 2\tilde{c}(J)$,
	it is sufficient that $(1-\kappa)^t\abs{J}\hmax\leq\tilde{c}(J)$.
	We have,
	\begin{align}
		(1-\kappa)^t\abs{J}\hmax\leq\tilde{c}(J) &\Longleftrightarrow
			t\log(1-\kappa)\leq \log\tilde{c}(J) - \log\abs{J} - \log\hmax \\
		&\Longleftrightarrow
			t \geq \frac{\log\abs{J} - \log\tilde{c}(J) + \log\hmax}{-\log(1-\kappa)} \\
		&\Longleftrightarrow
			t \geq a_0\log\frac{\abs{J}}{\tilde{c}(J)} + b_0
	\end{align}
	where $a_0\isdef -1/\log(1-\kappa)$ and $b_0\isdef -\log\hmax/\log(1-\kappa)$.
\end{proof}

The latter proposition can be interpreted as follows.
For $n\geq 0$, consider a hypercube $S_n\isdef\intinterval{0,n-1}^d$ of size $n^d$
in the lattice.
Then, $\tilde{c}(S_n)=\Theta(n^{d-1})$ as $n\to\infty$.
Thus, according to \cref{prop:entropy:evolution:equilibrium},
irrespective of the distribution of the initial configuration~$X^0$,
we have $H(X^t_{S_n})\geq n^d\hmax - \Theta(n^{d-1})$
(i.e., $S_n$ lacks no more than $\Theta(n^{d-1})$ nats of entropy at time $t$)
as soon as $t\geq \Theta(\log n)$.

\subsection{A bootstrap lemma}
\label{sec:entropy:bootstrap}

The next step is a ``bootstrap argument''.
The intuitive idea is as follows.  The effect of noise on a hypercube
$S_n\isdef\intinterval{0,n-1}^d$
is to accumulate entropy as long as the entropy of $S_n$ is less than
its maximum capacity $\abs{S_n}\hmax$.  A surjective CA on the other hand
keeps the entropy of $S_n$ almost preserved except for a leakage of
size $\bigo(\abs{\partial S_n})$ per iteration
through the boundary of $S_n$.  The ``equilibrium'' is reached at time $t_n=\bigo(\log n)$,
when the rate of accumulation and the maximum rate of leakage roughly match.
Now consider a much larger hypercube $S_m$ which contains many disjoint 
copies of $S_n$.
For $S_m$, a similar ``equilibrium'' is reached at time $t_m=\bigo(\log m)$.
However, the entropy leaking from the copies of $S_n$
will not have enough time to reach and escape through the boundary of $S_m$ before time~$t_m$,
and will hence have to accumulate inside $S_m$.  This implies that the entropy missing from each copy of $S_n$
at time $t_m$ must in fact be much less than $\Theta(\abs{\partial S_n})$.

Let us make this argument precise.
Given a random configuration $X$ and finite set $A\subseteq\ZZ^d$,
let us write $\Xi(X_A)\isdef\abs{A}\hmax-H(X_A)$
for the difference between the entropy of $X_A$ and the maximum entropy capacity of $A$.
Note that $\Xi(X_A)\geq 0$ with equality if and only if $X_A$ is uniformly distributed
over~$\Sigma^A$.
\begin{lemma}[Bootstrap lemma]
\label{lem:bootstrap}
	Let $\Phi$ be a PCA on $\Sigma^{\ZZ^d}$ with neighbourhood 
	$\Neighb\subseteq\intinterval{-r,r}^d$.
	Let $\tau,\delta:\ZZ^+\to[0,\infty)$
	be two functions satisfying the following property:
	\begin{itemize}
		\item for every trajectory $X^0,X^1,\ldots$ of $\Phi$ and each $n\in\ZZ^+$, 
			we have $\Xi(X^t_{S_n})\leq\delta(n)$ for all $t\geq\tau(n)$,
			irrespective of the distribution of $X^0$.
	\end{itemize}
	Let $k,m,n,t\in\ZZ^+$ be such that $m\geq k(n+2rt)$ and $t\geq\tau(m)$.
	Then, for every trajectory $X^0,X^1,\ldots$ of $\Phi$,
	we have $\Xi(X^t_{S_n})\leq\delta(m)/k^d$.
\end{lemma}
\begin{proof}
	Observe that we can pack $k^d$ disjoint copies of $\Moore^{rt}(S_n)$ in $S_m$. 
	Namely, for $w\in S_k$, let $Q_w\isdef(n+2rt)w+\intinterval{rt,rt+n-1}^d$.
	Then, the sets $\Moore^{rt}(Q_w)$ (for $w\in S_k$) are disjoint and are all included in $S_m$.
	Construct a random configuration $Y$ by choosing the patterns
	$Y_{\Moore^{rt}(Q_w)}$ (for $w\in S_k$) independently according to
	the distribution of $X^0_{\Moore^{rt}(S_n)}$, and assigning arbitrary values to the remaining cells.
	Consider a trajectory $Y^0,Y^1,\ldots$ of $\Phi$ with initial configuration $Y^0\isdef Y$.
	Clearly, the patterns $Y^t_{Q_w}$ (for $w\in S_k$)
	are independent and have the same distribution as $X^t_{S_n}$.
	Therefore, using the chain rule, we have
	\begin{align}
		H(Y^t_{S_m}) &=
			\sum_{w\in S_k} H(Y^t_{Q_w}) +
				H\bigg(Y^t_{S_m\setminus\bigcup_{w\in S_k}Q_w}\given[\bigg]Y^t_{\bigcup_{w\in S_k}Q_w}\bigg) \\
		&\leq
			k^d H(X^t_{S_n}) + \abs[\bigg]{S_m\setminus\bigcup_{w\in S_k}Q_w}\hmax \;,
	\end{align}
	which implies
	\begin{align}
		\Xi(Y^t_{S_m}) &\geq k^d \Xi(X^t_{S_n}) \;.
	\end{align}
	Now, since $Y^0,Y^1,\ldots$ is a trajectory of $\Phi$ and $t\geq\tau(m)$,
	we have $\Xi(Y^t_{S_m})\leq\delta(m)$.
	It follows that $\Xi(X^t_{S_n})\leq\delta(m)/k^d$, as claimed.
\end{proof}

\subsection{Proof of the theorem}
\label{sec:entropy:proof}

\begin{proof}[Proof of \cref{thm:main-result}]
	Let $X^0,X^1,\ldots$ be a trajectory of $\Phi$.
	For every finite set $A\subseteq\ZZ^d$,
	we show that $\Xi(X^t_A)=\abs{A}\hmax - H(X^t_A)\to 0$ exponentially fast as $t\to\infty$.
	This would imply that the distribution of $X^t$ converges weakly
	to the uniform Bernoulli measure $\lambda$.
	We then translate the bound on $\Xi(X^t_A)$ to a bound on total variation distance.
	
	Let $r\in\NN$ be such that the neighbourhood of $F$ (and hence also of~$\Phi$)
	is included in~$\intinterval{-r,r}^d$.
	Let $\tilde{c}(J)$ and $a_0,b_0>0$ be as in \cref{prop:entropy:evolution:equilibrium},
	and note that $\tilde{c}(S_n)=\Theta(n^{d-1})$ as $n\to\infty$.
	Choose constants $a_1,b_1,c_1>0$ such that
	$\tau(n)\isdef a_1\log n + b_1\geq a_0\log\frac{\abs{S_n}}{\tilde{c}(S_n)}+b_0$
	and $\delta(n)\isdef c_1 n^{d-1}\geq 2\tilde{c}(S_n)$
	for every $n\in\ZZ^+$.
	Then, by \cref{prop:entropy:evolution:equilibrium},
	for each $n\in\ZZ^+$ we have $\Xi(X^t_{S_n})\leq\delta(n)$ whenever $t\geq\tau(n)$,
	hence the hypothesis of \cref{lem:bootstrap} is satisfied.
	
	Suppose that $A\subseteq\ZZ^d$ is a finite set of cells
	with diameter~$n$.
	This means that $A\subseteq u+S_n$ for some $u\in\ZZ^d$.
	For $t\geq 0$, define $m_t\isdef k_t(n+2rt)$,
	where $k_t\in\ZZ^+$ is to be determined.
	Then,
	according to \cref{lem:bootstrap},
	\begin{align}
		\Xi(X^t_A) &\leq \frac{\delta(m_t)}{k_t^d}
			= \frac{c_1 k_t^{d-1} (n+2rt)^{d-1}}{k_t^d} = c_1 k_t^{-1}(n+2rt)^{d-1} \;,
	\end{align}
	provided that
	\begin{align}
	\label{eq:main-result:proof:t}
		t &\geq \tau(m_t) = a_1\log k_t + a_1\log(n+2rt) + b_1 \;.
	\end{align}
	
	Now, pick $\beta_1$ such that $0<\beta_1<1/a_1$,
	and set $k_t\isdef\lfloor\ee^{\beta_1 t}\rfloor$.
	Observe that with this choice,
	condition~\eqref{eq:main-result:proof:t} is satisfied
	for all sufficiently large~$t$.
	In particular, we can find constants $a,b>0$
	such that \eqref{eq:main-result:proof:t} holds whenever $t\geq a\log n + b$.
	It follows that, for a suitable constant $c_2>0$,
	$\Xi(X^t_A)\leq c_2\ee^{-\beta_1t}(n+2rt)^{d-1}$ for all $t\geq a\log n + b$.
	(Here, we need a new constant $c_2$ instead of $c_1$
	in order to compensate for replacing $\lfloor\ee^{\beta_1 t}\rfloor$ with $\ee^{\beta_1 t}$.)
	In particular, $\Xi(X^t_A)\to 0$ exponentially fast as $t\to\infty$.

	Next, let $\mu^t$ denote the distribution of $X^t$,
	and denote by $\mu^t_A$ and $\lambda_A$ the marginals of
	$\mu^t$ and $\lambda$ on~$A$.
	Observe that $\Xi(X^t_A)$ is the same as $D(\mu^t_A\|\lambda_A)$,
	the Kullback--Leibler divergence of $\mu^t_A$ relative to $\lambda_A$.
	According to Pinsker's inequality~\cite[Lemma~11.6.1]{CovTho06},
	we have
	\begin{align}
		\norm{\mu^t_A-\lambda_A}_\TV &\leq \sqrt{\frac{1}{2} D(\mu^t_A\|\lambda_A)}
			= \sqrt{\frac{1}{2}\Xi(X^t_A)} \;.
	\end{align}
	We conclude that
	\begin{align}
		\norm{\mu^t_A-\lambda_A}_\TV &\leq
			\sqrt{c_2/2}\,\ee^{-(\beta_1/2)t}(n+2rt)^{(d-1)/2}	
	\end{align}
	for every finite $A\subseteq\ZZ^d$ with diameter~$n$ and all $t\geq a\log n + b$.
	The claim follows by choosing $\beta>0$ slightly smaller than $\beta_1/2$
	and $\alpha>0$ sufficiently large.
\end{proof}

\section{Rapid information loss}
\label{sec:information-loss}
\subsection{Reversible CA on an infinite lattice}
\label{sec:information-loss:infinite-lattice}

\Cref{thm:main-result} shows that
for a surjective CA subject to positive additive noise, mixing occurs quite fast,
in the sense that it takes only $\bigo(\log n)$ steps before
the marginal on each hypercube of size $n^d$ is within $\varepsilon$-distance
from its stationary value.

To make this precise, let us use the notation
\begin{align}
	\norm{\mu-\nu}_A &\isdef \norm{\mu_A-\nu_A}_\TV
		= \frac{1}{2}\sum_{u\in\Sigma^A}\abs{\mu([u])-\nu([u])}
\end{align}
for the total variation distance between the marginals of two measures $\mu$ and $\nu$
on $A\Subset\ZZ^d$.  Let
\begin{align}
	d_A(t) &\isdef \sup_{x\in\Sigma^{\ZZ^d}}\norm{\Phi^t(x,\cdot)-\lambda}_A
\end{align}
be the maximum distance from stationarity of the marginal on $A$ at time $t$.
Note that $d_A(t)$ is non-increasing with $t$.
Given a finite set $A\subseteq\ZZ^d$ and $\varepsilon>0$, we let
\begin{align}
	t_\mix(A,\varepsilon) &\isdef \inf\{t: d_A(t)\leq\varepsilon\} \;.
\end{align}
We call $t_\mix(A,\varepsilon)$ the \emph{mixing time} of set $A$
at accuracy level $\varepsilon$ (cf.~\cite{LevPerWil09}).

As before, we let $S_n\isdef\intinterval{0,n-1}^d$ be a hypercube of size $n^d$
in the lattice.
\begin{corollary}[Mixing time of surjective CA with additive noise]
	Let $\Phi$ be a PCA describing the perturbation of a surjective CA
	$F:\Sigma^{\ZZ^d}\to\Sigma^{\ZZ^d}$ with positive additive noise.
	For every $\varepsilon>0$,
	we have $t_\mix(S_n,\varepsilon)=\bigo(\log n)$ as $n\to\infty$.
\end{corollary}
\begin{proof}
	According to \cref{thm:main-result},
	\begin{align}
		d_{S_n}(t) &\leq \alpha \ee^{-\beta t} n^{(d-1)/2}
	\end{align}
	for every $n\in\ZZ^+$ and $t\geq a\log n + b$.
	Therefore, $d_{S_n}(t)\leq\varepsilon$ as soon as
	\begin{align}
		t &\geq
			\max\left\{
				a\log n + b,
				\frac{d-1}{2\beta}\log n + \frac{1}{\beta}(\log\alpha - \log\varepsilon)
			\right\} \;,
	\end{align}
	which means $t_\mix(S_n,\varepsilon)=\bigo(\log n)$.
\end{proof}

In other words, for every $n>0$ and any accuracy $\varepsilon>0$,
the distribution of the pattern on $S_n$ becomes $\varepsilon$-indistinguishable
from the uniform distribution after $\bigo(\log n)$ time steps.

\subsection{Reversible parallel computers with finite space}
\label{sec:information-loss:finite-lattice}

In the proof of \cref{thm:main-result}, the bootstrap lemma was
needed to handle the potential diffusion of entropy on the infinite lattice.
For a reversible parallel computer with finite space, a simpler entropy
argument can be used to show that, in the presence of positive additive noise,
the state of the system becomes indistinguishable from pure noise
in logarithmic number of steps.
This is a reformulation of Theorem~2 of
Aharonov et al.~\cite{AhaBenImpNis96} in a slightly more general setup.

We consider a model of parallel reversible computation with finite space
in which every piece of data is (reversibly) processed at each time step,
and is hence exposed to noise.
More specifically, let $A$ be a finite set and $\Sigma$ a finite alphabet.
Let $\family{F}$ be a family of bijective maps $F:\Sigma^A\to\Sigma^A$.
Consider a computational process whose data is an element of $\Sigma^A$,
and in which, at every step, an arbitrary map from $\family{F}$ is applied to the data,
and the result is then subjected to positive additive noise.
For instance, a reversible logic circuit with $t$ layers, $n$ nodes per layer,
and wires between consecutive layers only,
in which every node is subject to positive noise fits in this setting.

The above process can be described by a time-inhomogeneous finite-state Markov chain
which mixes rapidly. 
To make this precise,
let us review some terminology and notation regarding finite-state Markov chains~\cite{LevPerWil09}.
Let $\Phi$ be an ergodic (possibly time-inhomogeneous)
Markov chain with finite state space $\pspace{X}$ and unique stationary distribution $\pi$.
We write $\Phi^{s\to t}$ for the transition matrix from time~$s$ to time~$t$.
Let
\begin{align}
	d_\Phi(t) &\isdef \sup_{x\in\pspace{X}}\norm{\Phi^{0\to t}(x,\cdot)-\pi}_\TV
\end{align}
denote the maximum distance from stationarity of the distribution at time $t$.
Since $\Phi$ is ergodic, $d(t)\to 0$ monotonically as $t\to\infty$.
The \emph{mixing time} at accuracy $\varepsilon>0$ is defined as
\begin{align}
	t_\mix(\Phi,\varepsilon) &\isdef \inf\{t: d_\Phi(t)\leq\varepsilon\} \;.
\end{align}

\begin{theorem}[Rapid mixing of reversible computer with noisy components]
\label{thm:reversible-computer:mixing-time}
	Let $A$ be a finite set and $\Sigma$ a finite alphabet.
	Let $\family{F}$ be a family of bijective maps $F:\Sigma^A\to\Sigma^A$.
	Let $\Phi$ be a time-inhomogeneous Markov chain on $\Sigma^A$ in which,
	at each time step, first an (arbitrary) element of $\family{F}$
	is applied to the current state, and then the state is subjected
	to positive additive noise with fixed noise distribution.
	Then, $\Phi$ is ergodic with the uniform distribution on~$\Sigma^A$
	as the stationary distribution.
	Furthermore, for every $\varepsilon>0$,
	$t_\mix(\Phi,\varepsilon)=\bigo(\log\abs{A})$ as $\abs{A}\to\infty$.
\end{theorem}
\begin{proof}
	Let $X^t$ denote the state of the Markov chain at time $t$.
	Let $\lambda$ denote the uniform distribution on $\Sigma^A$.
	Let $F_t\in\family{F}$ be the map applied at time~$t$,
	and let $q$ denote the noise distribution.
	Thus, $X^t$ is obtained from $X^{t-1}$ by applying additive noise with distribution $q$
	to $F_t(X^{t-1})$.
	Let $\kappa\isdef\abs{\Sigma}\min_{a\in\Sigma}q(a)$.
	
	The bijective maps $F_t$ do not change the entropy,
	hence according to \cref{lem:entropy:noise:joint},
	\begin{align}
		H(X^t) &\geq \abs{A}\kappa\hmax + (1-\kappa)H\big(F_t(X^{t-1})\big) \\
			&= \abs{A}\kappa\hmax + (1-\kappa)H(X^{t-1})
	\end{align}
	for each $t\geq 1$.  Therefore, setting $\Xi(X^t) \isdef \abs{A}\hmax - H(X^t)$,
	we have
	\begin{align}
		\Xi(X^t) &\leq (1-\kappa)^t\Xi(X^0) \leq (1-\kappa)^t\abs{A}\hmax
	\end{align}
	for every $t\geq 0$.
	This shows that the Markov chain is ergodic with $\lambda$
	as the unique stationary distribution.
	
	Now, recall that $\Xi(X^t)$ is the same as the Kullback--Leibler divergence
	$D(\mu^t\,\|\,\lambda)$, where~$\mu^t$ is the distribution of the Markov chain
	at time~$t$.
	Therefore, Pinsker's inequality~\cite[Lemma~11.6.1]{CovTho06}
	gives
	\begin{align}
		\norm{\mu^t-\lambda}_\TV &\leq \sqrt{\frac{1}{2} D(\mu^t\,\|\,\lambda)}
		=
			\sqrt{\frac{1}{2}\Xi(X^t_A)}
		\leq
			\sqrt{\hmax/2}\abs{A}^{\nicefrac{1}{2}}(1-\kappa)^{\nicefrac{t}{2}} \;.
	\end{align}
	Since $X^0$ is arbitrary, we get
	\begin{align}
		d_\Phi(t) &\leq \sqrt{\hmax/2}\abs{A}^{\nicefrac{1}{2}}(1-\kappa)^{\nicefrac{t}{2}} \;,
	\end{align}
	from which it follows that $t_\mix(\Phi,\varepsilon)=\bigo(\log\abs{A})$ as $\abs{A}\to\infty$.
\end{proof}

\section{Discussion}
\label{sec:discussion}

\paragraph*{Is there anything we can do with a noisy reversible CA?}

Since an input of size~$n$ is lost in $\bigo(\log n)$ steps,
no cell on the lattice will have time to ``sense'' all the input.
Thus, it seems unlikely that one can do any meaningful computation
with a noisy reversible CA in a scalable fashion.

Theorem~1 of Aharoni et al.~\cite{AhaBenImpNis96} states that,
with exponential redundancy, any logic circuit can be simulated
by a reliable reversible circuit.
In order to simulate a computation with $d$ steps (i.e., a circuit of depth~$d$),
every input bit is provided in $3^d$ separate copies.
The $i$-th step of the computation is performed $3^{d-i+1}$ times in parallel on separate copies,
and majority vote on groups of~$3$ is used to pass the results of the $i$-th step
to the next step.
Such an error-correcting mechanism cannot be implemented
in the setting of cellular automata on~$\ZZ^d$,
because the dependence graph of each cell grows at most polynomially.

\paragraph*{Reversible serial computers}

\Cref{thm:reversible-computer:mixing-time} describes the loss of information
in a parallel model of a reversible computer in which every bit of data is updated at every step
of the computation and is therefore subjected to noise.
In a serial computer (such as a Turing machine) on the other hand,
only a bounded portion of data is operated on at each step, and only that portion
is affected by significant noise.
We can there wonder about the possibility of having a reversible serial computer
that is capable of performing a significantly long computation reliably in the presence of noise.

\paragraph*{How much irreversibility is needed to do reliable computation with a CA?}

\Cref{thm:main-result} suggests that,
in order to perform reliable computation (with a CA-like computer),
some degree of irreversibility is unavoidable.
Can we quantify the degree of irreversibility needed to perform reliable computation
at a given noise level?

\bibliography{siamak-bibliography}

\end{document}